\documentclass[final]{article}

\usepackage[utf8]{inputenc}
\usepackage{mdframed}
\usepackage{graphicx}
\usepackage{mathtools}
\usepackage{amsmath}
\usepackage{amsthm}
\usepackage{color}
\usepackage{amssymb}
\usepackage{amsfonts}

\usepackage{algorithmic}
\usepackage{algorithm}
\usepackage{tcolorbox}

\newtheorem{theorem}{Theorem}

\newtheorem{claim}{Claim}

\newcommand{\abs}[1]{\left| #1 \right|}
\newcommand{\ckm}{\text{CkM}}
\newcommand{\ckmp}{\text{CkMP}}

\newcommand{\opt}{{\cal O}}

\newcommand{\opthat}{\hat{\cal O}}

\newcommand{\opthattwo}{{\cal O}'}
\newcommand{\optbar}{\bar{\cal O}}

\newcommand{\soln}{{\cal S}}

\newcommand{\solnhat}{\hat{\cal S}}
\newcommand{\solnhattwo}{{\cal S}'}
\newcommand{\solnbar}{\bar{\cal S}}

\newcommand{\cost}[1]{cost({#1})}

\newcommand{\sj}[1]{{\mathcal S}_{#1}}
\newcommand{\oj}[1]{{\mathcal O}_{#1}}

\newcommand{\fac}{{\cal F}}
\newcommand{\cli}{{\cal C}}

\newcommand {\bag}{{\cal B}}

\newcommand{\dist}[2]{c(#1,~#2)}

\newcommand{\bs}[1]{{\cal B}_{\cal S}({#1})}

\newcommand{\bo}[1]{{\cal B}_{\cal O}({#1})}

\newcommand{\bol}[1]{{\cal B}_{\cal O}^{L}({#1})}

\newcommand{\ball}[2]{{\cal B} ( {#1}, {#2} ) }
\newcommand{\balll}[2]{{\cal B}_L ( {#1}, {#2} ) }

\newcommand{\sigmas}[1]{\pi_{\soln}(#1)}
\newcommand{\sigmao}[1]{\pi_{\opt}(#1)}

\newcommand{\swapone}[2]{swap({#1},{#2})}
\newcommand{\swaptwo}[4]{double$-$swap(\{{#1},{#2} \},\{{#3},{#4}\})}

\newcommand{\swapmulti}[2]{\textit{multi-swap}({#1},{#2})}

\newcommand{\ie}{\textit{i}.\textit{e}.}
\newcommand{\facilityset}{\mathcal{F}}	

\newcommand{\capacity}{\textit{U}}

\newcommand{\light}{{\mathcal S_L}}
\newcommand{\heavy}{{\mathcal S_H}}
\newcommand{\good}{{\mathcal S_g}}
\newcommand{\bad}{{\mathcal S_b}}
\newcommand{\nice}{{\mathcal S_n}}
\newcommand{\sw}{{\mathcal S_W}}
\newcommand{\sB}{{\mathcal S_{\bag}}}

\newcommand{\op}{{\mathcal O}_{sp}}
\newcommand{\og}{{\mathcal O}_{g}}
\newcommand{\ob}{{\mathcal O}_{b}}
\newcommand{\on}{{\mathcal O}_{n}}
\newcommand{\ow}{{\mathcal O_W}}
\newcommand{\oB}{{\mathcal O_{\bag}}}

\newcommand{\tsoln}{{\mathcal T_{\soln}}}
\newcommand{\topt}{{\mathcal T_{\opt}}}

\newcommand{\meta}[1]{{#1}^{*}}

\newcommand{\vs}[1]{{v}({#1})}

\newcommand{\dom}[1]{{\mathcal O}_d(#1)}
\newcommand{\domtwo}[2]{{\mathcal O}_d(\{#1,#2\})}
\newcommand{\vstwo}[2]{{v}(\{#1,#2\})}

\newcommand{\floor}[1]{\bigl\lfloor#1\bigr\rfloor}

\newcommand{\sumlimits}[2]{\sum_{#1}^{#2}}

\newcommand{\etal}{\textit{et~al}.~}

\newcommand{\pj}[1]{{p}_{#1}}

\newcommand{\pencli}[1]{{\cal P}({#1})}

\newcommand{\doo}[1]{{D}({#1})}
\newcommand{\dop}[1]{{D'}({#1})}
\newcommand{\po}[1]{{P}({#1})}
\newcommand{\dpo}{{\mathcal D}_{sp}}

\newcommand{\ssp}{{\mathcal S}_{sp}}
\newcommand{\swaptwopo}[2]{double$-$swap({#1},{#2})}

\begin{document}
	
	\title{Improved Local Search Based Approximation Algorithm for Hard Uniform Capacitated k-Median Problems}
	
	\maketitle
	\begin{center}
		\author{Neelima Gupta$^1$,}
		\author{Aditya Pancholi$^2$}
	\end{center}
	\begin{enumerate}
		\item {Department of Computer Science, University of Delhi, India.\\
			\texttt{ngupta@cs.du.ac.in}}
		\item {Department of Computer Science, University of Delhi, India.\\
			\texttt{apancholi@.cs.du.ac.in}}
		%\authorrunning{S. Grover et. al.}
	\end{enumerate}

	\begin{abstract}
		
		In this paper, we study the hard uniform capacitated $k$- median problem 
		using local search heuristic. Obtaining a constant factor approximation for the problem is open. All the existing solutions giving constant-factor approximation, violate at least one of the constraints (\emph{cardinality}/ \emph{capacity}). All except Koruplou \etal \cite{KPR} are based on LP-relaxation. 
		
		We give $(3+\epsilon)$ factor approximation algorithm for the problem violating the cardinality by a factor of $8/3 \approx 2.67$.
		There is a trade-off between the approximation factor and the cardinality violation between our work and the existing work.
		Koruplou \etal \cite{KPR} gave $(1 + \alpha)$ approximation factor with $(5 + 5/\alpha)$ factor loss in cardinality using local search paradigm. Though the approximation factor can be made arbitrarily small, cardinality loss is at least $5$. 
		On the other hand, we improve upon the result of  Aardal \etal \cite{capkmGijswijtL2013} in terms of factor-loss. They gave $(7+\epsilon)$ factor approximation, with the cardinality violation by a factor $2$. 
		Most importantly, their result is obtained using LP-rounding, whereas local search techniques are straightforward, simple to apply and have been shown to perform well in practice via empirical studies.	
		
		We extend the result to hard uniform capacitated $k$-median with penalties.
		 To the best of our knowledge, ours is the first result for the problem.
		
	\end{abstract}

	\section{Introduction}
		$k$ - median problem is one of the extensively studied problem in literature {\cite{capkmGijswijtL2013, Archer03lagrangianrelaxation, kmedian, arya,Byrkaesa2015, charikar,Charikar:1999,jain2002new,  Jain:2001,li2013approximating}}.
		The problem is known to be NP-hard.
		The input instance consists of a set $\fac$ of facilities, a set $\cli$ of clients, a non-negative integer $k$  
		and a non-negative cost function $c$ defining cost to connect clients to the facilities. {Metric version of the problem assumes that $c$ is symmetric and satisfies triangle's inequality}. The goal is to select a subset $\soln \subseteq \fac$ as centers with $|\soln| \leq k$ (cardinality constraint) and to assign clients to them such that the total cost of serving the clients from centers is minimum. In {\em capacitated} version of the problem, we are also given a bound $u_i$ on the maximum number of clients that facility $i$ can serve.  
		{The {\em soft} capacity version allows a facility to be opened any number of times whereas the {\em hard} capacity version restricts the facilities to be opened at most once.} In $k$-median with penalties, each client $j$ has an associated penalty $p_j$ and we are allowed not to serve some clients at the cost of paying penalties for them. In this paper, we address the hard-capacitated $k$ median ($\ckm$) problem and its penalty variant($\ckmp$), when the capacities are uniform $\ie$ $u_i = \capacity$ for all $i \in \facilityset$. Our results are stated in Theorems~\ref{theo-ckm} and~\ref{theo-ckmpen}. For the problems, we define an $(a,b)$-approximation algorithm as a polynomial-time algorithm that computes a solution using at most $bk$ number of facilities with cost at most $a$ times the cost of an optimal solution using at most $k$ facilities.
		
		\begin{theorem}
			\label{theo-ckm}
			There is a polynomial time local search heuristic that approximates hard uniform capacitated $k$ median problem within $(3 + \epsilon)$ factor of the optimal violating the cardinality by a factor of $\frac{8}{3}$.
		\end{theorem}

	In contrast to the LP-based algorithms, local search technique is known to be straightforward , simple to apply and has been shown to perform well in practice via empirical studies \cite{KPR,KH:warehouse}. 
	Power of local search technique over the LP-based algorithms is well exhibited by the fact that there are constant factor approximation ($3$ for uniform and $5$ for non-uniform) \cite{Aggarwal,Bansal} for capacitated facility location problem whereas the natural LP is known to have an unbounded integrality gap.
	On the other hand,	local search heuristics are 
	notoriously hard to analyze. This is evident from the fact  that the only work known, based on local search heuristics, for $\ckm$, is due to Korupolu \etal \cite{KPR} more than $15$ years ago.

	Our work provides a trade-off between the approximation factor and the cardinality violation with the existing work. Koruplou \etal \cite{KPR} gave $(1 + \alpha)$ approximation factor with $(5 + 5/\alpha)$ factor loss in cardinality using local search paradigm. Though the approximation factor can be made arbitrarily small, cardinality loss is at least $5$. 
	Small approximation factor is obtained at a big loss in cardinality. 
	For example, for $\alpha$ anything less than $1$, cardinality violation is more than $10$.
	To achieve $3$ factor approximation using their heuristic, cardinality violation is $7.5$. Thus, we improve upon their result in terms of cardinality.
	On the other hand, we improve upon the results in \cite{capkmGijswijtL2013,capkmshili2014,Lisoda2016} in terms of factor-loss though the cardinality loss is a little more in our case. Aardal \etal \cite{capkmGijswijtL2013} gave $(7+\epsilon)$ factor approximation, with the violation of cardinality by a factor $2$ using LP Rounding. $O(1/\epsilon^2)$ factor approximation is given in \cite{capkmshili2014,Lisoda2016} violating the cardinality by a factor of $(1 + \epsilon)$ using sophisticated strengthened LPs.
	 
		\begin{theorem}
			\label{theo-ckmpen}
			There is a polynomial time local search heuristic that approximates hard uniform capacitated $k$ median problem with penalties within $(3 + \epsilon)$ factor of the optimal violating the cardinality by a factor of $\frac{8}{3}$.
		\end{theorem}

	To the best of our knowledge, no result is known for $\ckmp$. 	
	
		\subsection{Related Work}	
		Both, LP-based algorithms as well as local search heuristics, have been used to obtain good approximate algorithms for the (uncapacitated) $k$-median problem. {\cite{capkmGijswijtL2013, Archer03lagrangianrelaxation, kmedian, arya,Byrkaesa2015, charikar,Charikar:1999,jain2002new,  Jain:2001,li2013approximating}}. The best known factor of $2.611 + \epsilon$ was given by Byrka \etal \cite{Byrkaesa2015}. 
		Obtaining a constant approximation factor for $\ckm$ is an open problem. Natural LP is known to have an unbounded integrality gap when one of the constraints (cardinality/capacity) is allowed to be violated by a factor of less than $2$ without violating the other constraint, even for uniform capacities.

		Several constant factor approximations are known~\cite{capkmByrkaFRS2013,charikar,Charikar:1999,capkmshanfeili2014,Groveretal2016} for the problem that violate the capacities  by a factor of $2$ or more.  
		A $(7 + \epsilon)$ algorithm was given by Aardal \etal ~\cite{capkmGijswijtL2013} violating the cardinality constraint by a factor of $2$.
		Koruplou \etal \cite{KPR} gave $(1 + \alpha)$ approximation factor with $(5 + 5/\alpha)$ factor loss in cardinality.
		Very recently, Byrka \etal ~\cite{ByrkaRybicki2015} broke the barrier of $2$ in capacities and  gave an  $O(1/\epsilon^2)$ approximation violating capacities by a factor of $(1 + \epsilon)$ factor for uniform capacities.  For non-uniform capacities, a similar result has been obtained by Demirci \etal\  in ~\cite{Demirci2016}. 	
		Li~\cite{capkmshili2014,Lisoda2016} strengthened the LP to break the barrier of $2$ in cardinality and gave an $(O(1/\epsilon^2 \ log (1/\epsilon)))$ approximation using at most $(1 + \epsilon) k$ facilities. 
		Though the algorithm violates the cardinality only by $1 + \epsilon$, it introduces a softness bounded by a factor of $2$. The running time of the algorithm is $n^{O(1/\epsilon)}$.

		The other commonly used technique for the problem is local search \cite{kmedian,charikar,KPR} with the best factor of $3+ \epsilon$ given by Arya \etal \cite{kmedian}.	Local search technique has been particularly useful to deal with capacities for the facility location problem~\cite{Chudak, paltree,zhangchenye, gupta2008simpler, vygen,Aggarwal,Bansal}.
	
	Some results are known for the penalty variant of (uncapacitated) facility location problems, TSP and steiner network problems~\cite{charikar2001algorithms, Jain, xu2005lp,Xu2017,Goemans:1995:GAT:207985.207993, Bienstock1993}.
	For the capacitated variant of facility location problem with penalties, $5.83 + \epsilon$ factor approximation for uniform and $8.532 + \epsilon$ factor for non-uniform capacities were given by Gupta and Gupta in~\cite{guptangupta}. This is the only result known for the problems with extension on capacities as well as penalties.

	\subsection{High Level Idea}		
	 Let $\soln$ denote any feasible solution. The algorithm performs one of the following operations if it reduces the cost and it halts otherwise.
	 The local search operations are $\swapone{s}{o}; ~s\in \soln, o \in \fac\setminus\soln$ and $\swaptwo{s_1}{s_2}{o_1}{o_2}; ~s_1, s_2 \in \soln, o_1,o_2 \in \fac\setminus\soln$. Given a set of open facilities, min-cost flow problem is solved to obtain the optimal assignments of clients to opened facilities.
	 
	 To define the swaps and the reassignments for the purpose of analysis, we extend the ideas of  Arya \etal \cite{kmedian}. Swaps are defined so that every facility in optimal solution is swapped in at least once and at most thrice whereas facilities in our locally optimal solution is swapped out at most thrice.
	 When a facility in our locally optimal solution is swapped out, some of its clients are reassigned to other  facilities in our solution via a mapping similar to the one defined in \cite{kmedian}. However, for the capacitated case, mapping needs to be done a little carefully. An almost fully utilized facility may not be able to accommodate all the clients mapped to it and conversely a partially utilized facility may not be able to accommodate the load of an almost fully utilized facility. To address this concern, we partition the facilities of our locally optimal solution into {\em{heavy}} (denoted by $\heavy$) and {\em{light}} (denoted by $\light$). A facility is said to be {\em{heavy}}  if it serves more than ($3\capacity/5$) clients in our solution and is called {\em{light}}  otherwise. Heavy facilities neither participate in swaps nor the mapping.
	  Thus, mapping is defined  between the clients of light facilities only. We allow to open ($\frac{8}{3}k$) facilities in our solution so that we have at least $k$ light facilities.

	 There are two situations in which we may not be able to define a feasible mapping between the clients of two light facilities. First situation is explained as follows: Let $\opt$ denote some optimal solution, let $\mathcal{M}_o$ be the number of clients, a facility $o \in \opt$ shares with the light facilities of our solution. All the clients of a facility $s \in \light$ cannot be mapped to  clients of other facilities of our solution if $s$ shares more than $\mathcal{M}_o/2$ clients with $o$. Second situation arises when $s$ shares more than $2\capacity/5$ clients with $o$. In this case, mapping may be possible but it may not be feasible as the other facility $s'$, to which its clients are mapped, may not have sufficient available capacity to accommodate the clients of $s$. We say that $s$ {\em dominates} $o$ in the first case and that $s$ {\em covers} $o$ in the second case.

	 Although a facility $s \in \light$ may dominate several facilities in $\opt$, it can cover at most one facility in $\opt$. Whereas, a facility $o \in \opt$ can be dominated by at most one facility in $\light$, 
	 it can be covered by at most two facilities in $\light$. The scenario in which a facility $o\in \opt$ is covered by exactly $2$ facilities, say $s_1$ and $ s_2$ needs to be handled carefully. 
	  In this case, we say that $s_1$ as well as $s_2$ {\em{specially covers}} $o$. We denote the set of such facilities in $\opt$ as $\op$. Since mapping of clients of $s_1$ and $s_2$ cannot be done in $o$, we would like to swap $s_1$ and $s_2$ with $o$ and, assign their clients to $o$, i.e. we would like to perform $\swapone{\{s_1, s_2\}}{o}$. However since we do not have this operation, we look for one more facility $o'$ in $\opt$ so that we can perform double-swap of $\{s_1, s_2\}$ with $\{o, o'\}$. First we look for $o'$ such that $\{s_1,s_2\}$ together either dominate or cover $o'$. Clearly neither $s_1$ nor $s_2$, being light, can cover any facility other than $o$. Thus we look for $o'$ that is dominated by them. If $\{s_1, s_2\}$ do not dominate any facility other than $o$, we form a triplet $<s_1, s_2, o>$ and keep it aside. We call such triplets are nice triplets. They will be used to swap in some facilities of $\opt$ which are not swapped in otherwise.
	  If they dominate exactly one facility $o'$, then we perform double-swap of $\{s_1, s_2\}$ with $\{o, o'\}$. If they dominate at least two facilities other than $o$, then we cannot swap them out at all. We call such a pair of facilities as a bad pair.

	 Remaining facilities in $\light$ are classified as good, bad and nice. A facility that does not dominate any facility in $\opt$ is termed as {\em nice}. A nice facility can be swapped in with any facility in $\opt$. A facility $s$ that dominates exactly one facility $o$ in $\opt$ is termed as {\em good}. We perform (single) $swap(s, o)$ in this case.  A facility $s$ that dominates more than one facilities in $\opt$ is termed as {\em bad}. Bad facilities cannot participate in swaps. Let $\opthat$ denote the set of facilities of $\opt$ that are either dominated by bad facilities or by bad pairs in $\light$. Facilities of $\opthat$ are swapped in using the triplets (using $\swaptwo{.}{.}{.}{.}$ or the nice facilities (using $\swapone{.}{.}$). We show that the total number of triplets and the nice facilities is at least one thirds of $|{\opthat}|$ so that each facility of $\light$ is swapped out most $3$ times and each facility of $\opt$ is swapped in at least once and at most $3$ times (Note that in the process, the facilities of $\opt$ which were there in the triplets also get swapped thrice). Swapping in a facility of $\opt$ thrice contributes a factor of $3$ and swapping out a facility of $\light$ thrice contributes a factor of $6$ making a total of $9$ factor approximation.

	 Extending swap and double-swap to multi-swap, where upto $p, (p > 2)$ facilities can be swapped simultaneously, we are able to ensure that every $s \in \soln$ is swapped out at most $1 + 4/ (p-2)$ times,  and every $o \in \opt$ is swapped in at most $ 1 + 4/(p-2)$ times thereby reducing the factor to $(3+\epsilon)$.

	For $\ckmp$, we start with an initial feasible solution with $8k/3$ facilities from $\fac$. The clients are assigned by solving min cost flow problem over the facilities $\soln \cup \{\delta\}$, where $u_{\delta}=\abs{\cli}$ and $\forall j \in \cli, ~\dist{\delta}{j} = \pj{j}$. Clients assigned to $\delta$ pay penalty in the solution $\soln$. We bound the cost of the locally optimal solution, in the same manner as done for $\ckm$.

		\subsection{Local Search Paradigm} 
		
			Given a problem $\mathtt{P}$, local search algorithm starts with a candidate feasible solution $\soln$. A set of operations are defined such that performing an operation results in a new solution $\soln'$, called the neighbourhood solution of $\soln$. A solution $\soln$ may have more than one neighbourhood solutions. An operation is performed if it results in improvement in the cost.
			We formally describe the steps of the algorithm for a minimization problem.\\
		
			\textbf{The paradigm:} 
			\begin{enumerate}
				\item Compute an arbitrary feasible solution $\soln$ to $\mathtt{P}$.
				
				\item  \textbf{while} $\soln'$ is a neighborhood solution of $\soln$ such that $\cost{\soln'}<\cost{\soln}$ \\
				\textbf{do}
				$\soln \leftarrow \soln'$.
			\end{enumerate}
			
			The algorithm terminates at a locally optimal solution $\soln$, \ie   $\cost{\soln'}>\cost{\soln}$ for every neighborhood solution $\soln'$ 
			
			In the above algorithm presented, we move to a new solution if it gives some improvement in the cost, however small that improvement may be. This may lead to an algorithm taking lot of time. To ensure that the algorithm terminates in polynomial time, a local search step is performed only when the cost of the current solution $\soln$ is reduced by at least $\frac{\cost{\soln}}{p(n,\epsilon)}$, where $n$ is the size of the problem instance and $p(n,\epsilon)$ is an appropriate polynomial in $n$ and $1/\epsilon$ for a fixed $\epsilon > 0$. This modification in the algorithm incurs a cost of additive $\epsilon$ in the approximation factor.

	\subsection{Organization of the paper}
	  For the sake of easy disposition of ideas, we first present a weaker result for $\ckm$ in Section~\ref{sec-CkM-algo1}. The algorithm uses two operations: a (single) swap and a double swap and provides an ($9 + \epsilon, 8/3$) solution. The factor is subsequently improved to $(3 + \epsilon)$ in Section~\ref{sec-CkM-algo2} using multi-swap operation.
	  The results are then extended to $\ckmp$ in Section~\ref{sec-CkMP-algo}.

	\section{($9 + \epsilon, 8/3$) algorithm for Capacitated $k$-Median Problem}
		  \label{sec1} \label{sec-CkM-algo1}
		
		In this section, we present a local search algorithm that computes a solution with cost at most $9+\epsilon$ times the cost of an optimal.
		We start with an initial feasible solution selected as an arbitrary set of $8k/3$ facilities. Given a set of open facilities, optimal assignments of the clients is obtained by solving min-cost flow problem. 
		
		For any feasible solution $\soln$, algorithm performs one of the following operations, if it reduces the cost and terminates when it is no longer possible to improve the cost using these operations.
				
		\begin{enumerate}
		
				\item 
				$\swapone{s}{o}$: $\soln \leftarrow \soln \setminus \{s\} \cup \{o\}$, 
				%$\soln = \soln \setminus A \cup B$,
				$o \in \fac \setminus \soln$, 
				$s \in \soln$.

				\item 
				$\swaptwo{s_1}{s_2}{o_1}{o_2}$: $\soln \leftarrow \soln \setminus \{s_1, s_2\} \cup \{o_1, o_2\}$,
				%$\soln = \soln \setminus A \cup B$,
				$o_1, o_2 \in \fac \setminus \soln$, 
				$s_1, s_2 \in \soln$. 
		\end{enumerate}

		\begin{claim} 
		\label{suniono}
		For the locally optimal solution $\soln$, and optimal solution $\opt$ we have, 
		\begin{enumerate}
			\item $\cost{\soln \setminus \{s\} \cup \{o\}} \geq \cost{\soln}; \ \forall s\in \soln, o \in \opt$
			
			\item $\cost{\soln \setminus \{s_1, s_2\} \cup \{o_1, o_2\}} \geq \cost{\soln}; \ \forall s_1, s_2\in \soln, o_1, o_2 \in \opt$
		\end{enumerate}
	\end{claim}

	\begin{proof}
		The claims follow trivially when $\soln \cap \opt = \phi$ by the local optimality of $\soln$. Next, suppose $\soln \cap \opt \ne \phi$. Let $\{s_1, o_1\} \in \soln \cap \opt$. Then $\soln \setminus \{s_1\} \cup \{o_1\} = \soln \setminus \{s_1\}$ if $s_1 \ne o_1$ and it is $= \soln$ otherwise. Clearly the cost of assignment to facilities in $\soln \setminus \{s_1\}$ can not be smaller than the cost of assignment to facilities in $\soln$. 
		If $s_2 \in \soln \setminus \opt, o_2 \in \opt \setminus \soln$, then $\cost{\soln \setminus \{s_2\} \cup \{o_2 \}} \ge \cost{\soln}$ (by the argument of single swap) and $\cost{\soln \setminus \{s_1, s_2\} \cup \{o_1, o_2\}} = \cost{\soln \setminus \{s_1, s_2\} \cup \{o_2\}} \ge \cost{\soln \setminus \{ s_2\} \cup \{o_2\}}$. All the other cases can be argued similarly.
	\end{proof}

		\subsection{Notations} \label{noatations}
			 
			Let $\soln$ denote the locally optimal solution and $\opt$ denote an optimal solution to the problem. Let $\bs{s}$ be the set of clients served by $s \in \soln$ and $\bo{o}$ be the set of clients served by $o \in \opt$.
			Let $\ball{s}{o}$ denote the set of clients served by $s\in\soln$ and $o\in\opt$ \ie $\ball{s}{o} = \bs{s} \cap \bo{o}$.
			For a client $j$, let $\sigmas{j}$ and $\sigmao{j}$ denote the facilities serving $j$ in $\soln$ and $\opt$ respectively. Let $\sj{j}$ and  $\oj{j}$ denote the service costs paid by $j$ in $\soln$ and $\opt$ respectively. 

		Facilities in $\soln$ are partitioned into {\em{heavy}} ($\heavy$) and {\em{light}} ($\light$).
		A facility $s \in \soln$ is said to be heavy if $ \bs{s} > \frac{3}{5} \capacity$ and light otherwise.
		When a facility in our locally optimal solution is swapped out, some of its clients are reassigned to other  facilities in our solution via a mapping similar to the one defined in \cite{kmedian}.
		We may not be able to define a feasible mapping for the heavy facilities. Thus heavy facilities are never swapped out and no client is mapped onto them for reassignment. 
		Consider a facility $o \in \opt$, let 
		$\bol{o} = \bo{o} \cap \cup_{s \in \light} \bs{s}$. Let $\mathcal{M}_o = \abs{\bol{o}}$. 
		
		%We partition the set $\light$ into sets $\ssp$ and $\soln \setminus \ssp$. Facilities in $\ssp$ and $\soln \setminus \ssp$ will be handled in slightly different manners. In order to define the set $\ssp$, 
		We introduce two concepts important to define the swaps and the mapping.
		\begin{itemize}
			\item  
			A facility $s \in \light$ is said to {\em dominate} $o$, if $\ball{s}{o} > \mathcal{M}_o/2$. 
			Note that a facility $o \in \opt$ can be dominated by at most one facility in $s \in \light$ where as a facility $s \in \light$ can dominate any number of facilities. 
			Extending the definition to set $T$, we say that a set $T \subseteq \light$ {\em dominates} $o \in \opt$ if $\sum_{s \in T} \ball{s}{o} > \mathcal{M}_o/2$. Let $\dom{T}$ denote the set of facilities dominated by $T$.
			When $T = \{s\}$, slightly abusing the notation we use $\dom{s}$ instead of $\dom{\{s\}}$.
			
			\item A facility $s \in \light$ is said to {\em cover} $o \in \opt$, if $\ball{s}{o} > \frac{2}{5} \capacity$.
			Note that if $ s \in \light$ then it can cover at most one facility in $\opt$. Also a facility $o \in \opt$ can be covered by at most $2$ facilities in $\light$. 
			Extending the definition to set $T \subseteq \light$ {\em covers} $o \in \opt$ if $\sum_{s \in T}{} \ball{s}{o} > \frac{2}{5}\capacity$. Let $\vs{T}$ denote the set of facilities covered by $T$. Also we will use $\vs{s}$ instead of $\vs{\{s\}}$ when $T = \{s\}$. 
		\end{itemize}

		\subsection{Analysis: The Swaps}
		\label{sec-swaps}
		Consider a set of facilities in $\opt$ such that each of them is covered by exactly two light facilities. Let $\op$ denote the set of such facilities. 
		For $ o \in \opt$, a $1-1$ and onto mapping $\tau : \bol{o} \rightarrow \bol{o}$ can be defined such that the following claim holds,
				
		\begin{claim}
		\label{claim:strongclaim} \label{taumapping}
		For $s \in \light$ and $o\in \opt | o \notin \dom{s}$
		\begin{enumerate}
			\item \label{prop1}
			$\tau(\ball{s}{o}) \cap \ball{s}{o} = \phi$.
			
			\item \label{prop2}
			If $ o \notin \op $ then $\abs{\{j \in \ball{s}{o} : \tau(j) \in \ball{s'}{o}\}} \le \frac{2}{5}\capacity,~\forall s' \neq s$.
			
		\end{enumerate}	
		\end{claim}

				\begin{proof}
					$\tau$ can be defined as follows: Order the clients in $\bol{o}$ as $j_0, j_1,...,j_{\mathcal{M}_o-1}$ such that for every $s \in S$ with a nonempty $\ball{s}{o}$, the clients in $\ball{s}{o}$ are consecutive; that is, there exists $r,s,~ 0 \leq r \leq s \leq \mathcal{M}_o-1$, such that $\ball{s}{o} = \{j_r,...,j_s\}$. Define $\tau(j_p)=(j_q)$, where $q =(p +\floor{\mathcal{M}_o/2})~modulo~\mathcal{M}_o$. 
					
					We show that $\tau$ satisfies the claim.
					We prove (\ref{prop1}) using contradiction.
					Suppose if possible that both $j_p$, $\tau(j_p)=j_q \in \ball{s}{o}$ for some $s$, where $|\ball{s}{o}| \leq \mathcal{M}_o/2$. If $q = p+\floor{\mathcal{M}_o/2}$, then $|\ball{s}{o}| \geq q-p+1 =\floor{\mathcal{M}_o/2}+1 > \mathcal{M}_o/2$. If $q = p+\floor{\mathcal{M}_o/2} - \mathcal{M}_o$, then $|\ball{s}{o}| \geq p-q+1 = \mathcal{M}_o - \floor{\mathcal{M}_o/2}+1 > \mathcal{M}_o/2$. In either case, we have a contradiction, and hence mapping $\tau$ satisfies the claim.				
					
					For (\ref{prop2}), as $o \notin \op$, then at most one facility can cover $o$. If $\vs{s} = o$ then for all $s' \neq s, \ball{s'}{o} \leq \frac{2}{5} \capacity$. And if $\vs{s} \neq o$ then $\ball{s}{o} \leq \frac{2}{5} \capacity$. In either case the claim $\abs{\{j \in \ball{s}{o} : \tau(j) \in \ball{s'}{o}\}} \le \frac{2}{5}\capacity,~\forall s' \neq s$ holds true.
					
				\end{proof}

			Mapping $\tau$ is used to reassign the clients of a facility $s$ that is swapped out to other facilities $s' \in \light$. Claim (\ref{claim:strongclaim}.\ref{prop1}) ensures that if $s$ does not dominate $o$, then the client $j \in \balll{s}{o}$ is mapped to some $s' \neq s$, whereas claim (\ref{claim:strongclaim}.\ref{prop2}) ensures
			that if $o \notin \op$, then no more than $\frac{2}{5}k$ clients are mapped to $s'$. But if $o \in \op$ such that $\po{o}=\{s, s'\}$ , then more than $\frac{2}{5}k$ clients may get mapped to $s'$. This scenario poses a major challenge; thus facilities in $\op$ are considered separately while defining the swaps.

			For $o\in \op$, $\exists s_1,s_2 \in \light$ such that $s_1 \neq s_2$ and $\vs{s_1}=\vs{s_2}=o$.
			Consider a facility $o \in \op$, then let $\po{o}$ denote the set $\{s_1, s_2\}$ such that $s_1 \neq s_2, \vs{s_1}=\vs{s_2}=o$. Let $\ssp= \cup_{o \in \op} \po{o}$. Let $\doo{o} = \dom{\po{o}}$ and 		$\dop{o} = \doo{o}\setminus \{o\}$.
			Let $\dpo = \cup_{o \in \op} \dop{o}$. Figure \ref{fig0}(a) shows the relationship between $\ssp$ and $\dpo \cup \op$.
			The following claims hold.
		
				\begin{claim}
				\label{po_claim}
				\label{c3}
					$\forall o,o' \in \op, o \neq o'$ we have $\po{o} \cap \po{o'} = \phi $.
				\end{claim}
		 
		 \begin{proof}
		 			Suppose if possible $\po{o} \cap \po{o'} \neq \phi$. Let $s \in \po{o} \cap \po{o'}$. This implies $\ball{s}{o} \geq 2\capacity/5$ and $\ball{s}{o'} \geq 2\capacity/5$ which is a contradiction as $s\in\light$.
		 \end{proof}

		 		\begin{claim}
		 		\label{do_claim}
		 		\label{c4}
		 			$\forall o,o' \in \op, o \neq o$ we have $\doo{o} \cap \doo{o'} = \phi$.
		 		\end{claim}
		 		
				\begin{proof}
					Suppose if possible let $o_1 \in \doo{o} \cap \doo{o'}$. This implies $o_1 \in \dom{\po{o}}$ and $o_1 \in \dom{\po{o'}}$. This is a contradiction as $\po{o} \cap \po{o'} = \phi$ from claim \ref{po_claim} and $o_1$ cannot be dominated dominated by two disjoint set of facilities.
				\end{proof}	
				
			\begin{claim}
			\label{dsp_claim}
			\label{c5}
				$\dpo \cap \op =\phi$.
			\end{claim}
			\begin{proof}
				Suppose if possible let $o \in \dpo \cap \op$. As $o \in \op$, we have $o \in \dom{\po{o}}$. By the definition we have $o \notin \dop{o}$ thus $o \in \dop{o'}$ for some $o' \neq o, o' \in \op$. This implies $o \in \dom{\po{o'}}$ which is a contradiction as $\po{o} \cap \po{o'} = \phi$ using claim \ref{po_claim}.
		\end{proof}				 		
		
				We consider at-most $k$ swaps, satisfying the following properties.
				\begin{enumerate}
					\item Each $o \in \opt$ is considered in atleast one swap and at most three swaps.
									
					\item If $s \in \heavy$, $s$ is not considered in any swap operation.
							
					\item Each $s \in \light$ is considered in at most three swaps.				
				
					\item If $\swapone{s}{o}$ is considered then $\forall o' \neq o$; $o' \notin \dom{s}$ and  $\forall o' \in \op: o' \neq o$; $s \notin \po{o'}$.
					
					\item If $\swaptwo{s_1}{s_2}{o_1}{o_2}$ is considered then $\forall o' \neq o_1, o_2$;  $o' \notin \domtwo{s_1}{s_2}$ and 
					$\forall o' \in \op: o' \neq o_1, o_2$; $s_1, s_2 \notin \po{o'}$.
					
				\end{enumerate}

			Let $\sw \subseteq \light$ and $\ow \subseteq \opt$ denote the set of facilities that have participated in the swaps at any point of time. Initially $\sw = \ow = \phi$. While considering the facilities, we also maintain the sets $\solnhat \subseteq \light$ and $\opthat\subseteq \opt$; initially $\solnhat = \opthat = \phi$. 
			The facilities in $\solnhat$ will never participate in any swap. Facilities in $\opthat$ correspond to the facilities in $\solnhat$ in some way which will become clear when we define the swaps.
			We also maintain a set of triplets denoted by $\bag$ (will be defined shortly) and two sets $\sB$ and $\oB$ corresponding to $\bag$. All the three sets are empty initially. Throughout we maintain that $\sw, \sB, \solnhat$ are pairwise disjoint and $\ow, \oB, \opthat$ are pairwise disjoint.

			For $o_1 \in \op$ with $\po{o_1} = \{s_1, s_2\}$.
			\begin{enumerate}
				\item If $\abs{\doo{o_1}} = 1$ then $\doo{o_1} = \{o_1\}$. In this case we call $\po{o_1}$ a {\em nice pair}.
				Set $\bag = \bag \cup \{ < s_1, s_2, o_1 > \} $,
				$\oB = \oB \cup \doo{o_1}$, $\sB = \sB \cup \po{o_1}$. 
	
				\item If $\abs{\doo{o_1}} = 2$, let $\doo{o_1} = \{o_1, o_2\}$ %such that $o_2 \notin \op$. 
				%Note that $o_2 \notin \op$ as $ \forall o' \neq o_1; o' \notin \svs{{s_1},{s_2}}$ by the property (\ref{op_prop2}) of set $\op$.
				In this case we call $\po{o_1}$ a { \em good pair} and consider $\swaptwopo{\po{o_1}}{\doo{o_1}}$ which is nothing but
				$\swaptwo{s_1}{s_2}{o_1}{o_2}$.
				Set $\ow = \ow \cup \doo{o_1}$, $\sw = \sw \cup \po{o_1}$. 
				
				\item If $\abs{\doo{o_1}} > 2$ then we call $\po{o_1}$ a {\em bad pair}.
				Set $\opthat = \opthat \cup \doo{o_1}$, 
				$\solnhat = \solnhat \cup \po{o_1}$. That is, put the bad pairs in $\solnhat$ and the facilities dominated by them in $\opthat$.
				%As mentioned earlier $ \forall o' \in \domtwo{s_1}{s_2} \setminus \{o_1\}; \ o' \notin \op$.
				Note that the cardinality of $\solnhat$ increased by $2$ while cardinality of $\opthat$ increased by at least $3$.	\label{32case}		
			\end{enumerate}				 
				 				
			Figure \ref{fig1}(a) shows the partitions $\light$ and $\opt$ at this time. Let $\solnhattwo = \sw \cup \sB \cup \solnhat$ and $ \opthattwo = \ow \cup \oB \cup \opthat$. Note that $\solnhattwo=\ssp$ and $\opthattwo=\op \cup \dpo$.
			Also, clearly $\abs{\sw} = \abs{\ow}$, $\abs{\sB} = 2\abs{\oB}$ as for every facility added to $\oB$, two facilities are added to $\sB$. 
			and $3| {\solnhat}|\leq 2 |{\opthat}|$. The last claim follows as for every two facilities added to $\solnhat$, at least three facilities are added to $\opthat$.
			Next, we consider the facilities in $\opt \setminus \opthattwo$ and $\light \setminus \solnhattwo$.
			We say that a facility $s \in \light \setminus \solnhattwo$ is {\em good} if $\abs{\dom{s}} = 1$, {\em bad} if $\abs{\dom{s}} > 1$, else {\em nice} (i.e. $\abs{\dom{s}} = 0$). Let $\good, \bad$, and $\nice$ denote the set of good, bad and nice facilities respectively and, $\opt \setminus \opthattwo$ is partitioned into $\og, \ob$ and $\on$. Let $\og$ denote the set of facilities in $\opt \setminus \opthattwo$ captured by good facilities. Let $\ob$ denote the set of facilities in $\opt \setminus \opthattwo$ captured by bad facilities, and let $\on$ denote the set of facilities in $\opt \setminus \opthattwo$ not captured by any facility in $\light \setminus \solnhattwo$. Figure \ref{fig0}(b) shows the relationship between $\good, \og$ and $\bad, \ob$.	
		
				\begin{figure}[h!]
				\begin{center}
					\includegraphics[width=\columnwidth]{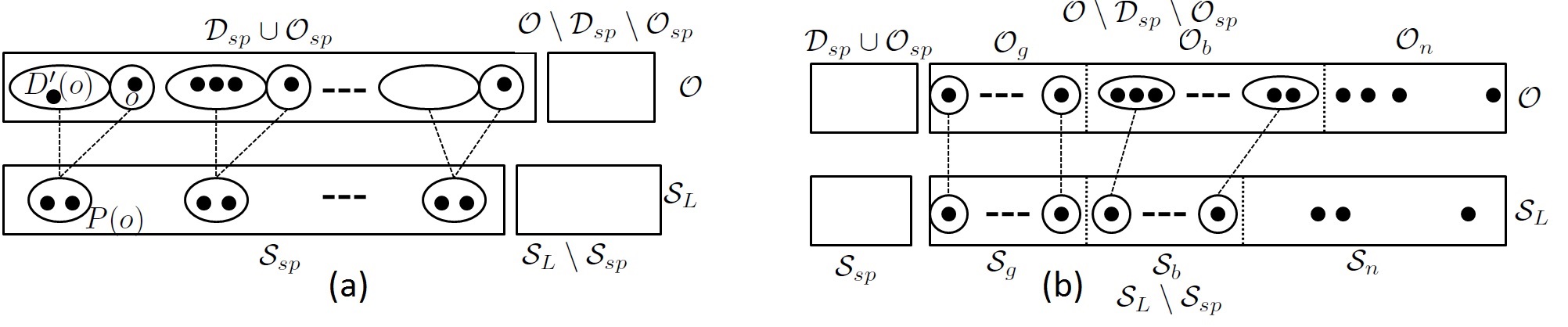}
					\caption{(a) Relationship between the partitions of $\ssp$ and $\dpo \cup \op$. (b)Relationship between the partitions of $\good$ and of $\og$, $\bad$.}
					% and $\ob$,$\nice$ and $\on$ .}
					\label{fig0}
				\end{center}
				\end{figure}	
		\begin{figure}[h!]
		\begin{center}
			\includegraphics[width=0.8\columnwidth]{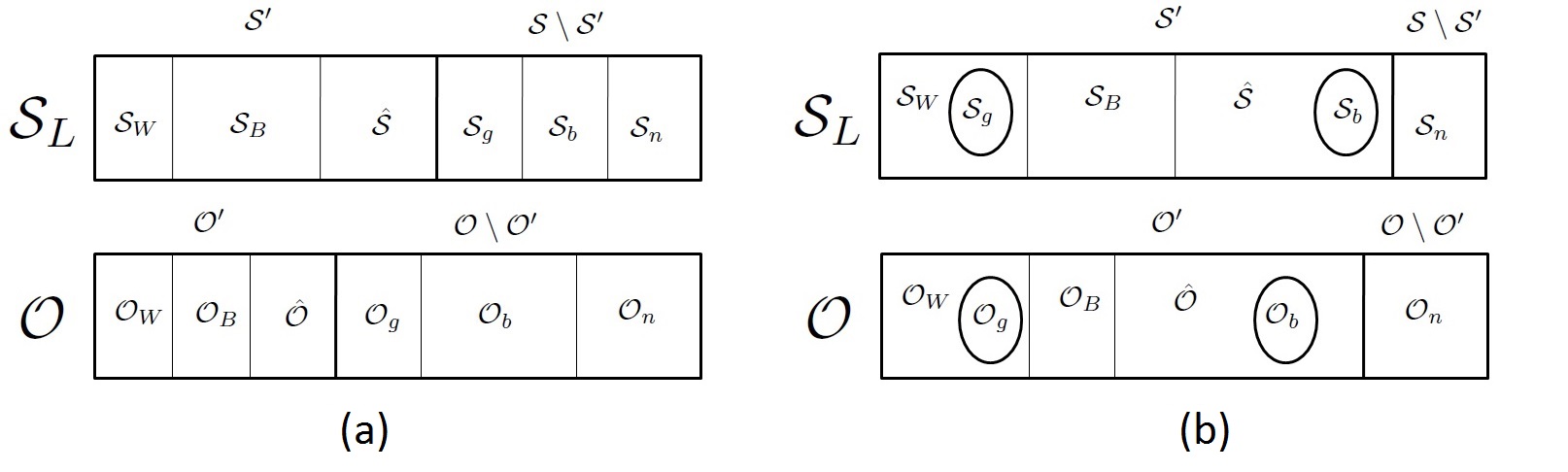}
			\caption{Partitions of $\light$ and $\opt$.}
			\label{fig1}
		\end{center}
		\end{figure}		
		
\begin{enumerate}
	\item For every $s \in \nice$ ($\abs{\dom{s}} = 0$): So nothing.
	\item For every $s \in \good$ ($ \abs{\dom{s}} = 1$): Perform $\swapone{s}{\dom{s}}$. %, where $s: \dom{s} =  o$ . 
	Update $\sw, \ow, \opthattwo$ and $\solnhattwo$ as $\sw = \sw \cup \good, \ow = \ow \cup \og, \opthattwo = \opthattwo \cup \ow$, $\solnhattwo = \solnhattwo \cup \sw$ in this order.
	
	\item For every $s \in \bad$ ($\abs{\dom{s}} = 0$): Set  $\solnhat = \solnhat \cup \{s\}, \opthat = \opthat \cup \dom{s}, \opthattwo = \opthattwo \cup \opthat$, $\solnhattwo = \solnhattwo \cup \solnhat$ in this order.
	That is, put the bad facility in $\solnhat$ and the facilities dominated by it in $\opthat$.
	%As mentioned earlier $ \forall o' \in \domtwo{s_1}{s_2} \setminus \{o_1\}; \ o' \notin \op$.
	%Note that 
	The cardinality of $\solnhat$ increased by $1$ while the cardinality of $\opthat$ increased by at least $2$.
\end{enumerate}

			New partitions are shown in Figure \ref{fig1}(b).
			%Figures \ref{fig0}(b) and \ref{fig1}(b). 
			Let $\solnbar$ be the set of facilities in $\light$ that have not participated in any swap. Then such a facility is either a nice facility, a bad facility,  is in a nice pair or in a bad pair. 
			Similarly, let $\optbar$ be the set of facilities in $\opt$ that have not participated in the above swaps.
			Then, $\oB$ is the set of facilities in $\optbar$ that  are in a triplet. Let $T = \optbar \setminus \oB$.
			Then facilities in $T$ are either dominated by a bad facility, by a bad pair, or are not dominated by any facility or a pair.$\ie$, $T = \opthat \cup \on$. Let $\ell$ be the number of such facilities $\ie$, $\ell = \abs{T}$. Next claim shows that there are at least $\ell/3$ nice facilities and nice pairs taken together.
		
			\begin{claim} \label{numsuff}
				$\abs{\bag} + \abs{\nice} \geq \frac{1}{3} ( \abs{\opthat} + \abs{\on})$
			\end{claim}	
		
			\begin{proof}
				While handling $\ob$, for every facility added in $\solnhat$  atleast $2$ facilities are added in $\opthat$ and while handling $\op$,  atleast $3$ facilities are added in $\opthat$ for every $2$ facilities added in $\solnhat$.  Thus we have $3 \cdot \abs{\solnhat} \leq 2 \cdot \abs{\opthat} 
				\implies \abs{\solnhat} \leq \frac{2}{3} \cdot\abs{\opthat}$
				
	%			\begin{equation}
	%			3.\abs{\solnhat} \leq 2.\abs{\opthat} 
	%			\implies \abs{\solnhat} \leq \frac{2}{3}.\abs{\opthat}
	%			\end{equation}
				
				Also $\abs{\opt} = k \leq \abs{\light}$. 
				$\opt = \opthattwo \cup \on$ and $\opthattwo \cap \on = \phi$
				
				Thus $\opt = \ow \cup \oB \cup \opthat \cup \on$
				
				Similarly $\light = \solnhattwo \cup \nice$ and $\solnhattwo \cap \nice = \phi$
				
				Thus $\light = \sw \cup \sB \cup \solnhat \cup \nice$
				
				Also $\abs{\ow} = \abs{\sw}$, $\abs{\oB} = \abs{\bag}$, $\abs{\sw} =2\abs{\bag}$ and $\abs{\solnhat} \leq \frac{2}{3}.\abs{\opthat}$ 
				
				Thus we get
				
				$\abs{\sw} + \abs{\solnhat} + 2\abs{\bag} + \abs{\nice} \geq \abs{\ow} + \abs{\bag} + \abs{\opthat} + \abs{\on}$
				
				$\abs{\bag} + \abs{\nice} \geq  \abs{\opthat} - \abs{\solnhat} + \abs{\on} \geq \frac{1}{3} \abs{\opthat} + \abs{\on}$
				
				$\abs{\bag} + \abs{\nice} \geq \frac{1}{3}( \abs{\opthat} + \abs{\on})$
				
				%$\implies \frac{1}{3}(\abs{\opthat} + \abs{\on}) \leq \frac{1}{3}\abs{\opthat} + \abs{\on} \leq \abs{\bag} + \abs{\nice}$

			\end{proof}

			Next, consider the following $\ell$ swaps in which the facilities in $\optbar$ are swapped with nice facilities or nice pairs in $\light$ in a way that each nice facility or a facility in a nice pair is considered in at most $3$ swaps and each facility in $\oB$ is also considered in at most $3$ swaps.
		\begin{enumerate}	
			\item Repeat until $\abs{T} < 3$. Pick $o_1, o_2, o_3 \in T$. 
			\begin{enumerate}
				\item If $\nice \neq \phi$. Pick a facility $s_1 \in \nice$; perform\\ $\swapone{s_1}{o_1}$, $\swapone{s_1}{o_2}$, $\swapone{s_1}{o_3}$.
				Set $\nice = \nice \setminus \{s_1\}$.
				
				\item Else, pick a triplet $<s_1, s_2, o> \in \bag$, and perform\\
				$\swaptwo{s_1}{s_2}{o}{o_1}$, 
				$\swaptwo{s_1}{s_2}{o}{o_2}$, 
				$\swaptwo{s_1}{s_2}{o}{o_3}$.
				\item Set $ T = T \setminus \{o_1, o_2, o_3\}$
			\end{enumerate} 
		\item If $\abs{T} > 0$, either there must be a facility $s_1 \in \nice$ or a triplet $<s_1, s_2, o> \in \bag$; accordingly perform swap or double-swap with the facilities in $T$ in the same manner as described in step 1.
			\end{enumerate} 
		
			The swaps are summarized in Figure \ref{fig2} and Figure \ref{fig3}.
				\begin{figure}[h!]
					\begin{center}
						\includegraphics[width=0.6\columnwidth]{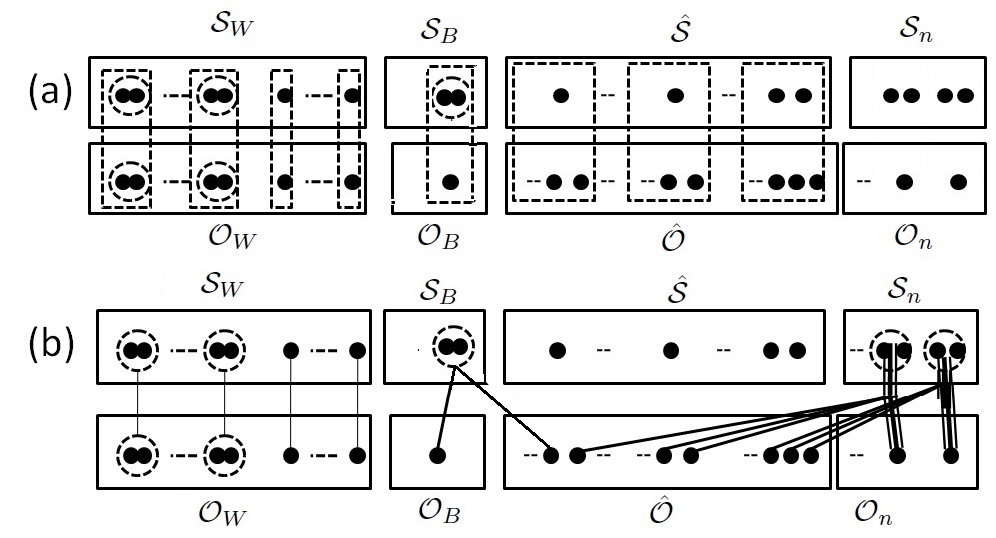}
						\caption{(a) Partitions of $\light$ and $\opt$ in terms of $\sw,\sB,\solnhat,\nice$ and $\ow, \ow, \opthat, \on$ respectively.
						(b) Swaps}
						\label{fig3}
					\end{center}
				\end{figure}
		
				\begin{figure}[h!]
				\begin{center}	
		
				\begin{tcolorbox}[colback=white] %\label{box1}
					\begin{enumerate}
						
						\item 
								For all $o_1 \in \op$ {\bf do} let $\po{o_1} = \{s_1, s_2\}$.
								\begin{enumerate}
									\item If $\abs{\doo{o_1}} = 1$ then $\doo{o_1} = \{o_1\}$.  $\po{o_1}$ is a nice pair.
									Set $\bag = \bag \cup \{ < s_1, s_2, o_1 > \} $,
									$\oB = \oB \cup \doo{o_1}$, $\sB = \sB \cup \po{o_1}$. 
						
									\item If $\abs{\doo{o_1}} = 2$, let $\doo{o_1} = \{o_1, o_2\}$ %such that $o_2 \notin \op$. 
									%Note that $o_2 \notin \op$ as $ \forall o' \neq o_1; o' \notin \svs{{s_1},{s_2}}$ by the property (\ref{op_prop2}) of set $\op$.
								$\po{o_1}$ is a good pair. Consider $\swaptwopo{\po{o_1}}{\doo{o_1}}$. %Which is nothing but
								%	$\swaptwo{s_1}{s_2}{o_1}{o_2}$.
									Set $\ow = \ow \cup \doo{o_1}$, $\sw = \sw \cup \po{o_1}$. 
									
									\item If $\abs{\doo{o_1}} > 2$ then 
									set $\opthat = \opthat \cup \doo{o_1}$, 
									$\solnhat = \solnhat \cup \po{o_1}$.
								\end{enumerate}

						\item For every $o \in \og$, consider $\swapone{s}{o}$, where $s: \dom{s} = o$. Update $\sw, \ow, \opthattwo$ and $\solnhattwo$ as $\sw = \sw \cup \good, \ow = \ow \cup \og, \opthattwo = \opthattwo \cup \ow$, $\solnhattwo = \solnhattwo \cup \sw$ in this order.
						
						\item Next consider $\ob$ and $\bad$. 
						Update $\solnhat, \opthat, \opthattwo$ and $\solnhattwo$ as $\solnhat = \solnhat \cup \bad, \opthat = \opthat \cup \ob, \opthattwo = \opthattwo \cup \opthat$, $\solnhattwo = \solnhattwo \cup \solnhat$ in this order.
						
						\item Let $T = \opthat \cup \on$.  Consider the following swaps:
						\begin{enumerate}	
							\item Repeat until $\abs{T} < 3$. Pick $o_1, o_2, o_3 \in T$. \label{step1}
							\begin{enumerate}
								\item If $\nice \neq \phi$. Pick a facility $s_1 \in \nice$; perform\\ $\swapone{s_1}{o_1}$, $\swapone{s_1}{o_2}$, $\swapone{s_1}{o_3}$
								Set $\nice = \nice \setminus \{s_1\}$.
								
								\item Else, pick a triplet $<s_1, s_2, o> \in \bag$, and perform\\
								$\swaptwo{s_1}{s_2}{o}{o_1}$, 
								$\swaptwo{s_1}{s_2}{o}{o_2}$, 
								$\swaptwo{s_1}{s_2}{o}{o_3}$.
								\item Set $ T = T \setminus \{o_1, o_2, o_3\}$
							\end{enumerate} 
							\item If $\abs{T} > 0$, either there must be a facility $s_1 \in \nice$ or a triplet $<s_1, s_2, o> \in \bag$; accordingly perform swap or double-swap with the facilities in $T$ in the same manner as described in step~\ref{step1}.
						\end{enumerate}

					\end{enumerate}
				\end{tcolorbox}
					\caption{Summary of the swaps}
							\label{fig2}
						\end{center}
						\end{figure}

	\subsection{Analysis: Bounding the Cost}
	\label{cost-bound}
		Now we bound the cost of these swaps. 
		Whenever we consider a swap of form $\swaptwo{s_1}{s_2}{o_1}{o_2}$, the mapping $\tau$ as defined in claim~\ref{suniono} cannot be used to reassign the clients of $s_1$ and $s_2$ as it is possible that 
		%for some $\tilde{o} \in \opt$, 
		some client $j$ of $s_1$ is mapped to a client of $s_2$ or vice versa.
		%some client $j \in \ball{s_1}{\tilde{o}}$ is getting mapped to some client $j' \in \ball{s_2}{\tilde{o}}$ or vice versa.
		To address this, we define another mapping $\tau'$ in a similar way as in claim (\ref{claim:strongclaim}) considering $s_1$ and $s_2$ as a single facility. 
			Let $DS = \{\{s_1,s_2\}: \swaptwo{s_1}{s_2}{.}{.}$
			was~performed $\}$.
		$\tau'$ satisfy the following claim
		
		\begin{claim}
		\label{claim:strongclaim1}
		\begin{enumerate}
			\item For $s \in \light$ and $o\in \opt | o \notin \dom{s}$
					\begin{enumerate}
						\item \label{prop01}
						$\tau'(\ball{s}{o}) \cap \ball{s}{o} = \phi$.
						
						\item \label{prop02}
						If $ o \notin \op $ then $\abs{\{j \in \ball{s}{o} : \tau'(j) \in \ball{s'}{o}\}} \le \frac{2}{5}\capacity,~\forall s' \neq s$.
					\end{enumerate}	
			\item For {$\{s_1,s_2\} \in DS$} and $o\in \opt | o \notin \domtwo{s_1}{s_2}$
			\begin{enumerate}
			\item \label{prop11}
				$\tau'(\ball{\{s_1, s_2\}}{o}) \cap \ball{\{s_1, s_2\}}{o} = \phi$.
						
			\item \label{prop21}
				If $ o \notin \op $ then $\abs{\{j \in \ball{\{s_1, s_2\}}{o} : \tau'(j) \in \ball{s'}{o}\}} \le \frac{2}{5}\capacity,~\forall s' \neq s_1, s_2$.
			\end{enumerate}
		\end{enumerate}		
		\end{claim}		
		
			\begin{proof}
				$\tau'$ can be defined as follows. Consider $s_1$ and $s_2$ as a single meta-facility, and define the mapping just as earlier. That is no $j \in \ball{s_1}{\tilde{o}}$ get mapped to $j' \in \ball{s_2}{\tilde{o}}$ and vice versa. Note that this is possible to create such a $\tau'$ as $\{s_1, s_2\}$ together do not dominate $\tilde{o}$.
				
				For (\ref{prop11}), as $o \notin \domtwo{s_1}{s_2}$; same argument ensures that $\tau'(\ball{\{s_1, s_2\}}{o}) \cap \ball{\{s_1, s_2\}}{o} = \phi$ still hold true  
				For (\ref{prop21}), as $o \notin \op$, then at most one facility can cover $o$. If $\vstwo{s_1}{s_2} = o$ then for all $s' \neq s, \ball{s'}{o} \leq \frac{2}{5} \capacity$. And if $\vstwo{s_1}{s_2} \neq o$ then $\ball{\{s_1,s_2\}}{o} \leq \frac{2}{5} \capacity$. In either case the claim $\abs{\{j \in \ball{\{s_1,s_2\}}{o} : \tau'(j) \in \ball{s'}{o}\}} \le \frac{2}{5}\capacity,~\forall s' \neq s_1,s_2$ holds true.			
			\end{proof}		
			
			Set $\tau = \tau'$. For $\swapone{s_1}{o_1}$ reassignment is done as follows: for  $j \in \bo{o_1}$ assign $j$ to $o_1$ and for $j \in \bs{s_1} \setminus \bo{o_1}$, assign $j$ to $\sigmas{\tau(j)}$. The cost of the operation is 		
		
			For $j \in \bs{s_1} \setminus \bo{o_1}, \sigmas{\tau(j)} \ne \sigmas{j}$. Thus,  since $j$ was assigned to $s$ and not to $\sigmas{\tau(j)}$ (call it $s'$), $\dist{j}{s} \le \dist{j}{s'}$ (since $s'$, being a light facility, had sufficient room to accommodate $j$ ). Thus,  $\dist{j}{s} \leq \dist{j}{s'} \leq \dist{j}{o'} + \dist{o'}{\tau(j)} + \dist{\tau(j)}{s'}$ or $\sj{j} \leq (\oj{j} + \oj{\tau(j)} + \sj{\tau(j)})$. However, if $j \in \bs{s_1} \cap \bo{o_1}, \sigmas{\tau(j)}$ may be same as $\sigmas{j}$. $\sj{j} \leq (\oj{j} + \oj{\tau(j)} + \sj{\tau(j)})$ follows trivially in this case by triangle inequality.
			Thus we can write

			%Note that $(\oj{j} + \oj{\tau(j)} + \sj{\tau(j)} - \sj{j}) >0$. As $j$ was assigned to $s$ and not $\sigmas{\tau(j)}$ (lets call it $s'$) even when $s'$ being a light facility has sufficient capacity to accommodate $j$, this mean $\dist{j}{s} \leq \dist{j}{s'} \leq \dist{j}{o'} + \dist{o'}{\tau(j)} + \dist{\tau(j)}{s'}$ or $\sj{j} \leq (\oj{j} + \oj{\tau(j)} + \sj{\tau(j)})$. In case if $\sigmas{j} = \sigmas{\tau(j)}$, this the above claim holds true by triangle's inequality.
			%Thus we can write
			\begin{equation} \label{eq:2.2}
				\sumlimits{j \in \bo{o_1}}{}(\oj{j} - \sj{j}) +  \sumlimits{j \in \bs{s_1}}{} (\oj{j} + \oj{\tau(j)} + \sj{\tau(j)} - \sj{j}) > 0
			\end{equation}

			For $\swaptwo{s_1}{s_2}{o_1}{o_2}$, reassignment is defined as, assign $j \in \bo{o_1}$ to $o_1$, assign $j \in \bo{o_2}$ to $o_2$, and assign $j \in \bs{s_1} \cup \bs{s_2} \setminus \bo{o_1} \setminus \bo{o_2}$ to $\sigmas{\tau^*(j)}$. The cost of the operation is 
					
			\begin{multline} \label{eq:2.3}
				\sumlimits{j \in \bo{o_1} \cup \bo{o_2}}{}(\oj{j} - \sj{j}) +  \\
				\sumlimits{j \in \bs{s_1} \cup \bs{s_2} \setminus \{\bo{o_1} \cup \bo{o_2}\}}{} (\oj{j} + \oj{\tau(j)} + \sj{\tau(j)} - \sj{j}) > 0
			\end{multline}
			Also, by similar arguement as above, we have
			
			\begin{multline} \label{eq:2.4}
				\sumlimits{j \in \bo{o_1} \cup \bo{o_2}}{}(\oj{j} - \sj{j}) +  \\
				\sumlimits{j \in \bs{s_1} \cup \bs{s_2}}{} (\oj{j} + \oj{\tau(j)} + \sj{\tau(j)} - \sj{j}) > 0
			\end{multline}

			Now, a facility $ o \in \opt$ may be considered at max three times. Let $\opt_1, \opt_2, \opt_3$ denote the set of facilities in $\opt$ that are considered in $1,2$ and $3$ swaps respectively.A facility $s \in \light$ may be swapped out at most $3$ times. Thus,
			we can write
			\begin{multline} \label{eq:2.5.1}
				\sumlimits{o \in \opt_1}{}\sumlimits{j \in \bo{o}}{}(\oj{j} - \sj{j}) + 
				\sumlimits{o \in \opt_2}{}\sumlimits{j \in \bo{o}}{}2(\oj{j} - \sj{j}) +
				\sumlimits{o \in \opt_2}{}\sumlimits{j \in \bo{o}}{}3(\oj{j} - \sj{j}) +\\		
				\sumlimits{s \in \light}{}\sumlimits{j \in \bs{s}}{} 3.(\oj{j} + \oj{\tau(j)} + \sj{\tau(j)} - \sj{j}) > 0
			\end{multline}
			as $\oj{j} + \oj{\tau(j)} + \sj{\tau(j)} - \sj{j} > 0$ as argued above.
			
		\begin{multline} \label{eq:2.5.2}
		\sumlimits{o \in \opt_1}{}\sumlimits{j \in \bo{o}}{}(\sj{j}) + 
		\sumlimits{o \in \opt_2}{}\sumlimits{j \in \bo{o}}{}2(\sj{j}) +
		\sumlimits{o \in \opt_2}{}\sumlimits{j \in \bo{o}}{}3(\sj{j})		
		<	\\	
		\sumlimits{o \in \opt_1}{}\sumlimits{j \in \bo{o}}{}(\oj{j}) + 
		\sumlimits{o \in \opt_2}{}\sumlimits{j \in \bo{o}}{}2(\oj{j}) +
		\sumlimits{o \in \opt_2}{}\sumlimits{j \in \bo{o}}{}3(\oj{j}) +\\		
		\sumlimits{s \in \light}{}\sumlimits{j \in \bs{s}}{} 3.(\oj{j} + \oj{\tau(j)} + \sj{\tau(j)} - \sj{j}) 
		\end{multline}
		
		\begin{multline} \label{eq:2.5.3}
		\sumlimits{o \in \opt}{}\sumlimits{j \in \bo{o}}{}(\sj{j})		
		<
		\sumlimits{o \in \opt}{}\sumlimits{j \in \bo{o}}{}3(\oj{j}) +		
		\sumlimits{s \in \light}{}\sumlimits{j \in \bs{s}}{} 3.(\oj{j} + \oj{\tau(j)} + \sj{\tau(j)} - \sj{j}) 
		\end{multline}
		
		\begin{multline} \label{eq:2.5.4}
		\sumlimits{j \in \cli}{}(\sj{j})		
		<
		\sumlimits{j \in \cli}{}3(\oj{j}) +		
		\sumlimits{s \in \light}{}\sumlimits{j \in \bs{s}}{} 3.(\oj{j} + \oj{\tau(j)} + \sj{\tau(j)} - \sj{j}) 
		\end{multline}
		\begin{multline} \label{eq:2.5.5}
		\cost{\soln}		
		<
		3\cost{\opt} +		
		\sumlimits{s \in \light}{}\sumlimits{j \in \bs{s}}{} 3.(\oj{j} + \oj{\tau(j)} + \sj{\tau(j)} - \sj{j})
		\end{multline}
		%Now, every facility in $\opt$ is considered in exactly one swap, so first term gets added over all $o \in \opt$, while 
		%Whereas a facility $s \in \light$ gets swapped out at most $3$ times. So we get
		%\begin{multline} \label{eq:2.5}
		%	\sumlimits{o \in \opt}{}\sumlimits{j \in \bo{o}}{}(\oj{j} - \sj{j}) + 
		%	\sumlimits{s \in \light}{}\sumlimits{j \in \bs{s}}{} 3.(\oj{j} + \oj{\tau(j)} + \sj{\tau(j)} - \sj{j}) > 0
		%\end{multline}
		
		%Considering all the swaps, we will get
		%\begin{multline} \label{eq:1.3}
		%\sumlimits{j \in \cli}{}(\oj{j} - \sj{j})  
		%+ 
		%\sumlimits{j \in \cli}{}
		%3(\oj{j} + \oj{\tau(j)} + \sj{\tau(j)} - \sj{j}) 
		%> 0
		%\end{multline}
		%where first term gives $\cost{\opt} - \cost{\soln}$ %while second term gives $6.\cost{\opt}$. Thus we get $\cost{\opt}-\cost{\soln} + 6\cost{\opt} \geq 0$. Thus, $\cost{\soln} \leq 7\cost{\opt}$.
		
		And as $\tau$ is a $1-1$ and onto mapping, 
		$$\sumlimits{s \in \light}{}\sumlimits{j \in \bs{s}}{}\oj{j} =\sumlimits{s \in \light}{}\sumlimits{j \in \bs{s}}{}\oj{\tau(j)}$$ 
		and $$\sumlimits{s \in \light}{}\sumlimits{j \in \bs{s}}{}(\sj{\tau(j)} - \sj{j}) = 0$$. Thus,
		\begin{multline}
		\sumlimits{s \in \light}{}\sumlimits{j \in \bs{s}}{} 3.(\oj{j} + \oj{\tau(j)} + \sj{\tau(j)} - \sj{j}) = 2 \sumlimits{s \in \light}{}\sumlimits{j \in \bs{s}}{} 3.(\oj{j}) = 6\sumlimits{s \in \light}{}\sumlimits{j \in \bs{s}}{} (\oj{j}) \leq\\ 6\sumlimits{s \in \soln}{}\sumlimits{j \in \bs{s}}{} (\oj{j}) = 6 \cost{\opt}
		\end{multline}

		%And as $\tau$ is a $1-1$ and onto mapping, $\sumlimits{j \in \cli}{}\oj{j} =\sumlimits{j \in \cli}{}\oj{\tau(j)}$ and $\sumlimits{j \in \cli}{}(\sj{\tau(j)} - \sj{j}) = 0$. Thus, $3(\oj{j} + \oj{\tau(j)} + \sj{\tau(j)} - \sj{j}) = 6\cost{\opt}$. 
		Thus we have
		
		\begin{multline} \label{eq:2.5.6}
		\cost{\soln}		
		<
		3\cost{\opt} +		
		\sumlimits{s \in \light}{}\sumlimits{j \in \bs{s}}{} 3.(\oj{j} + \oj{\tau(j)} + \sj{\tau(j)} - \sj{j}) \\
		\leq 3\cost{\opt} + 6\cost{\opt} = 9\cost{\opt} 
		\end{multline}

			\section{$(3+ \epsilon, 8/3)$ algorithm using multi-swaps} \label{sec-CkM-algo2}
			In this section, we reduce the factor to $(3+ \epsilon)$ using multi-swap operation.
			%extend the idea to multi-swaps. We present an algorithm that where upto $p > 1$ facilities can be swapped simultaneously. 
			Let $p = 2p'+2$ for some integer $p'$. Let $p > 2$. The algorithm performs the following operation if it reduces the cost of the solution and it terminates otherwise.
			%The swap operation can be generalized as
			$$\swapmulti{A}{B}: \soln = \soln \setminus B \cup A; B \subseteq \soln, A \subseteq \fac, \abs{A} = \abs{B} \leq p$$
			The operation can be performed in $O(n^p)$ time. For a fixed $\epsilon > 0$ and $p = O(1/\epsilon)$, it runs in $O(n^{1/\epsilon})$ time.

		\subsection{Defining the swaps}
		
	Following swaps are considered. Partitions are formed and the facilities in $\op$ and $\ow$ are treated in the same manner as described in Section~\ref{sec-swaps}.	
		%
	%	Just as earlier, consider the partitioning of $\light$ into $\sw, \sB, \solnhat, \good, \bad, \nice$ such that $\sw \cup \sB \cup \solnhat \cup \good \cup \bad \cup \nice = \light$. Also $\opt$ is participated into $\ow, \oB, \opthat, \og, \ob, \on$ such that $\ow \cup \oB \cup \opthat \cup \og \cup \ob \cup \on = \opt$. 
	%	
	We write inequalities corresponding to different operations and take their weighted sums.
	%We take the weighted sum of inequalities that we obtain from different operations.
	 Inequalities corresponding to the swaps defined for facilities in $\op$ and $\ow$ are assigned weight $1$.
	 
		Let $\topt = \oB \cup \opthat \cup \ob \cup \on$ and 
		$\tsoln = \sB \cup \solnhat \cup \bad \cup \nice$. $\tsoln$ and $\topt$ are partitioned into $A_1, A_2, \ldots, A_r$ and $B_1, B_2, \ldots, B_r$ respectively using the Partition Algorithm. %\ref{algo1} 
		such that the partitions satisfy the following properties:	
		\begin{enumerate}
			\item For $1 \leq i \leq (r-1)$, we have $\abs{A_i} = \abs{B_i}, B_i = \dom{A_i}$ and $\abs{A_r} = \abs{B_r}$.
			
			\item For $1 \leq i \leq (r-1)$, the set $A_i$ has exactly one facility $s$ from $\bad$ or exactly one bad pair $<s_1, s_2>$ from $\solnhat$. 
			%such that $\{s_1, s_2\} = \svsi{o}$ for some $o \in \op$....Aditya ...Can you please see this? Is it correct?..shouldn't it have been like bad pair? \textcolor{red}{Check this}
			
			\item The set $A_r$ contains  facilities only from $\sB \cup \nice$.
		\end{enumerate} 
	%	
		% $\abs{A_i} = \abs{B_i}$ and each $A_i$ contains either exactly one facility from $\bad$ or exactly one pair $<s_1, s_2>$ such that $<s_1, s_2, .> \in \bag$.
		
	\begin{algorithm}
	%	\label{algo1}
	\footnotesize
		\begin{algorithmic}[1]
			%\label{algo1}	
			%\Procedure{Meta-cluster}{$\mathcal{T}$}
			\STATE $i=0$.
			\WHILE {$\exists$ a facility in $\solnhat \cup \bad$}
			\STATE $i = i+1$
			\STATE $A_i = \{s\}$ where $s \in \bad$ or $A_i = \{s_1, s_2\}$ where 
			%$s_1, s_2 \in \solnhat$ such that 
			$\{s_1, s_2\}$ is a bad pair from $\solnhat$. 
			% and $\abs{\domtwo{s_1}{s_2}} >2$. \textcolor{red}{check this}
			\STATE $B_i = \dom{A_i}$ 
			
			\WHILE {$\abs{A_i} \neq \abs{B_i}$}
			\STATE $A_i = A_i \cup \{g\}$ where $g \in \tsoln \cap \{\sB \cup \nice \} \setminus A_i$.
			\STATE $B_i = \dom{A_i}$.
			\ENDWHILE
			
			\STATE $\topt = \topt \setminus B_i$, $\tsoln = \tsoln \setminus A_i$. 
			\ENDWHILE 
			\STATE $A_r = \tsoln$, $B_r = \topt$.
			%\EndProcedure
		\end{algorithmic}
		\label{alg1}
		\caption{Partition Algorithm}
	\end{algorithm}

	%	\textbf{procedure} partition:
	%	\begin{enumerate}
	%	\item $i=0$
	%	\item \textbf{while} $\exists$ a facility in $\solnhat \cup \bad$ \textbf{do}
	%	\begin{enumerate}
	%		\item $i = i+1$
	%		\item $A_i = \{s\}$ where $s \in \bad$ or $A_i = \{s_1, s_2\}$ where $<s_1, s_2, . > \in \bag$. % \in \solnhat$ such that $\abs{\domtwo{b_1}{b_2}} >2$.
	%		\item $B_i = \dom{A_i}$ 
	%		\item \textbf{while} $\abs{A_i} \neq \abs{B_i}$ \textbf{do}.
	%		
	%		\begin{enumerate}
	%			\item $A_i = A_i \cup \{g\}$ where $g \in \tsoln \cap \{\sB \cup \nice \} \setminus A_i$
	%			
	%			\item $B_i = \dom{A_i}$
	%		\end{enumerate}
	%		
	%		\item $\topt = \topt \setminus B_i$, $\tsoln = \tsoln \setminus A_i$.
	%		
	%	\end{enumerate}
	%	
	%	\item $A_r = \tsoln$, $B_r = \topt$.
	%	\end{enumerate}
	%	
	%	Now the sets $A_1, \ldots, A_r$ and $B_1, \ldots, B_r$ follow following properties.
	%	
	%	\begin{enumerate}
	%		\item For $1 \leq i \leq (r-1)$, we have $\abs{A_i} = \abs{B_i}, B_i = \dom{A_i}$ and $\abs{A_r} = \abs{B_r}$.
	%		
	%		\item For $1 \leq i \leq (r-1)$, the set $A_i$ has exactly one facility from $ \sB \cup \bad$. 
	%		
	%		\item The set $A_r$ contains only facilities from 
	%	\end{enumerate} 
	
		Next, we define the following swaps:
	%	
	%	We define the swaps as follows
		\begin{enumerate}
	%\item Like before, for every $\{s_1, s_2\} \in \sw; ~ \{o_1, o_2\} \in \ow$ such that $\domtwo{s_1}{s_2} = \{o_1, o_2\}$. 
	%	We consider
	%	$\swaptwo{s_1}{s_2}{o_1}{o_2}$.
		
	%	\item Like before, for every $s_1\in \good; ~ o_1 \in \ow$ such that $\dom{s_1} = o_1$. 
	%	We consider
	%	$\swapone{s_1}{o_1}$.
		
		\item For the sets $A_i, B_i$ for some $i \leq i \leq (r-1)$, 
		such that $\abs{A_i} = \abs{B_i} \leq p$, perform $\swapmulti{A_i}{B_i}$ :
		 for all $o \in B_i$, $j \in \bo{o}$, reassign $j$ to $o$, 
		 for all $j \in (\cup_{s \in A_i} \bs{s}) \setminus (\cup_{o \in B_i} \bo{o})$, assign $j$ to $\sigmas{\tau(j)}$. Inequalities are assigned weight $1$.

		\item For the sets $A_i, B_i$ for some $i \leq i \leq (r-1)$, such that $\abs{A_i} = \abs{B_i} = q > p$, we perform \textit{shrinking} as follows:
	%	
		%\begin{enumerate}
	%	\item
			for every $<s_1, s_2, o_1> \in \bag$ such that $s_1, s_2 \in A_i$ and $o_1 \in B_i$ \textbf{do}, $A_i = A_i \setminus \{s_1, s_2\}$, $B_i = B_i \setminus {o_1}$, $A_i = A_i \cup \meta{s_1}$. where $\meta{s_1}$ denote a meta node corresponding to nodes $s_1, s_2$.
			%....Aditya123s...can we have some better notation for this, m in the superscript is looking weird with no connection...may be just s-hat-one ....think about it. %We call this as \textit{shrinking}.
			 Note that after shrinking we still have $\abs{A_i} = \abs{B_i} = q' \geq q/2 $. Also as $q > p = 2p' +2$, we have $ q' > p'+1 = p/2$. We consider the swaps as follows
			\begin{enumerate}
			\item If $s \in A_i~ : s \in \bad$ then we consider exactly $q'(q'-1)$ swaps as follows: for all $o \in B_i$,
			for all $s' \in A_i \setminus \{s\}$,  if $s' = \meta{s_j}$ was  created by shrinking of $<s_{j_1}, s_{j_2}, o_{j_1}> $, then we perform $\swaptwo{s_{j_1}}{s_{j_2}}{o_{j_1}}{o}$.
				If $s'$ was not created by shrinking, then we perform $\swapone{s'}{o}$.
			
			Each inequality is assigned weight $1/(q'-1)$. Then, each $s' \in A_i$ participates in exactly $(q') / (q'-1) = 1 + 1/(q'-1) \leq 1+ 1/p' = 1+2/(p-2)$ swaps where as each $o \in B_i$ participate in exactly $1$ swap except for facilities of the type $o_{j_1}$ which participates in as many swaps as $s_{j_1}, s_{j_2}$ do, which is no more than $1+ 2/(p-2)$.   
			
			\item If $s_1, s_2 \in A_i~ : \{s_1, s_2\}$ is a bad pair from $\solnhat$ then we consider exactly $q'(q'-2)$ swaps as follows: for all $o \in B_i$,
			for all $s' \in A_i \setminus \{s_1, s_2\}$,  if $s' = \meta{s_j}$ was  created by shrinking of $<s_{j_1}, s_{j_2}, o_{j_1}> $, then we perform $\swaptwo{s_{j_1}}{s_{j_2}}{o_{j_1}}{o}$.
			If $s'$ was not created by shrinking, then we perform $\swapone{s'}{o}$.
			Each inequality is assigned weight $1/(q'-2)$. Then, each $s' \in A_i$ participates in exactly $(q') / (q'-2) = 1 + 2/(q'-2) \leq 1+ 2/p' = 1+4/(p-2)$ swaps where as each $o \in B_i$ participate in exactly $1$ swap except for facilities of the type $o_{j_1}$ which participates in as many swaps as $s_{j_1}, s_{j_2}$ do, which is no more than $1+ 4/(p-2)$.
			\end{enumerate} 
		
	%	\end{enumerate}
		 
		\end{enumerate}
	
		The cost of the swaps can be analysed as follows. Every $s \in \soln$ is swapped out at most $1 + 4/ (p-2)$ times and every $o \in \opt$ is swapped in atleast once and at most $ 1 + 4/(p-2)$ times. Thus we can now write

			\begin{multline} \label{eq:3.5.4}
					\sumlimits{j \in \cli}{}(\sj{j})		
					<
					(1 + 4/(p-2))\sumlimits{j \in \cli}{} (\oj{j}) +		
					(1 + 4/(p-2))\sumlimits{s \in \light}{}\sumlimits{j \in \bs{s}}{} (\oj{j} + \oj{\tau(j)} + \sj{\tau(j)} - \sj{j}) 
				\end{multline}
			which gives
			
			\begin{multline} \label{eq:3.5.5}
							\sumlimits{j \in \cli}{}(\sj{j})		
							<
							(1 + 4/(p-2))\cost{\opt} +		
							(1 + 4/(p-2))(2\cost{\opt}) = (3 + 12/(p-2))\cost{\opt}
						\end{multline}

			\section{Capacitated $k$-Median with Penalties}
			\label{sec-CkMP-algo}
				$\ckmp$ is a variation of $\ckm$ where we also have a fixed penalty cost $\pj{j}$ associated with each client $j \in \cli$. For a solution $\soln$, let $\pencli{\soln}$ denote the set of clients that pay penalties in $\soln$. The clients $\cli \setminus \pencli{\soln}$ are serviced by the facilities opened in $\soln$. We borrow the notations from Section \ref{sec1}. Then,  
			$\cost{\soln}$ is  $\sumlimits{s \in \soln}{} \sumlimits{j \in \bs{s}}{} \sj{j} + \sumlimits{j \in \pencli{\soln}}{} \pj{j}$. For easy disposition of ideas, we give $(9 + \epsilon, 8/3)$ algorithm and its analysis. The factor is reduced to $(3 + \epsilon)$ in exactly the same manner as is done in Section~\ref{sec-CkM-algo2}.
			
			%As $\soln$ is locally optimal
			
			\subsection{$(9+\epsilon, 8/3)$ Algorithm}
			
	%		\textbf{High Level Idea:} Given and instance $\inst= \instval$ of $\ckmp$, we start with an initial feasible solution with $8k/3$ facilities from $\fac$. The clients are assigned by solving min cost flow problem over the facilities $\soln \cup \{\delta\}$, where $u_{\delta}=\abs{\cli}$ and $\forall j \in \cli, ~\dist{\delta}{j} = \pj{j}$. Clients assigned to $\delta$ pay penalty in the solution $\soln$. We bound the cost of the locally optimal solution, in the way similar to \ckm.
			
		%	\textbf {Algorithm:} 
			We start with an initial feasible solution $\soln_0$ such that $\abs{\soln} = 8k/3$. 
			The clients are assigned by solving min cost flow problem over the facilities $\soln_0 \cup \{\delta\}$, where $u_{\delta}=\abs{\cli}$ and $\forall j \in \cli, ~\dist{\delta}{j} = \pj{j}$. Clients assigned to $\delta$ pay penalty in the solution $\soln_0$. 
			The operations available to the algorithm are $\swapone{s}{o}$ and $\swaptwo{s_1}{s_2}{o_1}{o_2}$ 
			%similar to section \ref{Unit1}.
			Let $\soln$ be the locally optimal solution with $\cost{\soln} = \sumlimits{s \in \soln}{} \sumlimits{j \in \bs{s}}{} \sj{j} + \sumlimits{j \in \pencli{\soln}}{} \pj{j}$.
			% Now we bound the cost of the solution $\soln$.

			\subsection{Analysis}
			Let $\opt$ be an optimal solution for the problem with $\cost{\opt} = \sumlimits{o \in \opt}{} \sumlimits{j \in \bo{o}}{} \oj{j} + \sumlimits{j \in \pencli{\opt}}{} \pj{j}$.	
			Swaps are defined exactly in the same manner as done in Section~\ref{sec1}. However, there is a slight change in the re-assignment.
			Consider the $\swapone{s}{o}: s \in \light$ and $o \in \opt$; reassignments are done as follows: $\forall j \in \bo{o}$, assign $j$ to $o$; $\forall j \in \bs{s} \cap \pencli{\opt}$, $j$ pays penalty $\pj{j}$ and $\forall j \in \bs{s} \setminus \{\bo{o} \cup \pencli{\opt} \}$, $j$ is assigned to $\sigmas{\tau(j)}$. %where the mapping $\tau$ is defined as same as in \ref{mapping}. 
			As $\soln$ is locally optimal, we have

			\begin{multline} \label{eq0.1}
			%	\cost{\swapone{s}{o}} =
				\sumlimits{j \in \bo{o} \setminus \pencli{\soln} }{}(\oj{j} - \sj{j}) +
				\sumlimits{j \in \bo{o} \cap 
					\pencli{\soln}}{}(\oj{j} - \pj{j}) +
			\sumlimits{j \in \bs{s} \cap \pencli{\opt}}{}(\pj{j} - \sj{j}) +	\\
				\sumlimits{j \in \bs{s} \setminus \bo{o} \setminus \pencli{\opt}}{} (\oj{j} + \oj{\tau(j)} + \sj{\tau(j)} - \sj{j}) > 0
			\end{multline}
	%		Note that $\{\bs{s} \setminus \bo{o} \} \cap \pencli{\opt}$ is same as $\bs{s} \cap \pencli{\opt}$. Thus the third term can be simplified to $\sumlimits{j \in \bs{s} \cap \pencli{\opt}}{}(\pj{j} - \sj{j})$.
			
			Similarly for $\swaptwo{s_1}{s_2}{o_1}{o_2}$, we have
			\begin{multline} \label{eq0.2}
				%\cost{\swaptwo{s_1}{s_2}{o_1}{o_2}} =
				\sumlimits{j \in \bo{o_1}\cup \bo{o_2}}{}(\oj{j} - \sj{j}) +
				\sumlimits{j \in \{\bo{o_1}\cup \bo{o_2} \} \cap 
					\pencli{\soln}}{}(\oj{j} - \pj{j}) +\\
				\sumlimits{j \in \{ \{\bs{s_1}\cup \bs{s_2} \} \cap 
					\pencli{\opt}}{}(\pj{j} - \sj{j}) +	\\
				\sumlimits{j \in \{\{\bs{s_1}\cup \bs{s_2} \} \setminus \{\bo{o_1}\cup \bo{o_2} \} \}\setminus \pencli{\opt} \}}{} (\oj{j} + \oj{\tau(j)} + \sj{\tau(j)} - \sj{j}) > 0
			\end{multline}		
	%		The third term can be simplified to $\sumlimits{j \in \{ \{\bs{s_1}\cup \bs{s_2} \} \cap \pencli{\opt}}{}(\pj{j} - \sj{j})$.
	%		
	%		As before a facility $ o \in \opt$ may be swapped in at most three times. $\opt_1, \opt_2, \opt_3$ denote the set of facilities in $\opt$ that are considered in $1,2$ and $3$ swaps respectively as earlier. A 
		and a	facility $s \in \light$ may be swapped out at most $3$ times. Thus summing over all swaps we have 	
			\begin{multline} \label{eq0.3}
				\sumlimits{o \in \opt_1}{}(\sumlimits{j \in \bo{o} \setminus \pencli{\soln}}{}(\oj{j} - \sj{j}) +
				\sumlimits{j \in \bo{o} \cap 
					\pencli{\soln}}{}(\oj{j} - \pj{j})) +\\
				\sumlimits{o \in \opt_2}{}(\sumlimits{j \in \bo{o} \setminus \pencli{\soln}}{} 2(\oj{j} - \sj{j}) +
				\sumlimits{j \in \bo{o} \cap 
					\pencli{\soln}}{} 2 (\oj{j} - \pj{j}) )+\\
				\sumlimits{o \in \opt_3}{}(\sumlimits{j \in \bo{o} \setminus \pencli{\soln}}{} 3(\oj{j} - \sj{j}) +
				\sumlimits{s \in \light}{}\sumlimits{j \in \bo{o} \cap 
					\pencli{\soln}}{} 3 (\oj{j} - \pj{j})) +\\
				\sumlimits{s \in \light}{}\sumlimits{j \in \bs{s} \cap 
					\pencli{\opt}}{} 3 (\pj{j} - \sj{j})+
				\sumlimits{s \in \light}{}\sumlimits{j \in \bs{s}\setminus \pencli{\opt}}{}  3(\oj{j} + \oj{\tau(j)} + \sj{\tau(j)} - \sj{j}) > 0
			\end{multline}		
			
			Note that for all $s \in \soln$, $\pj{j} - \sj{j} \ge 0$  for all $j \in \bs{s} \cap \pencli{\opt}$ otherwise $j$ would have paid penalty in $\soln$ too. Thus we can write $\sumlimits{s \in \light}{}\sumlimits{j \in \bs{s} \cap \pencli{\opt}}{} 3 (\pj{j} - \sj{j}) \leq \sumlimits{s \in \soln}{}\sumlimits{j \in \bs{s} \cap \pencli{\opt}}{} 3 (\pj{j} - \sj{j})$
			
			Rearranging, we get
			\begin{multline} \label{eq0.4}
				\sumlimits{o \in \opt_1 \cup \opt_2 \cup \opt_3}{}\sumlimits{j \in \bo{o} \setminus \pencli{\soln}}{} \sj{j}   +
				\sumlimits{o \in \opt_1 \cup \opt_2 \cup \opt_3}{}\sumlimits{j \in \bo{o} \cap 
					\pencli{\soln}}{}\pj{j} + 
				\sumlimits{j \in \pencli{\opt} \setminus \pencli{\soln}}{} \sj{j} \le\\
				\sumlimits{o \in \opt_1 \cup \opt_2 \cup \opt_3}{}\sumlimits{j \in \bo{o} \setminus \pencli{\soln}}{} 3\oj{j} 
				+ \sumlimits{o \in \opt_1 \cup \opt_2 \cup \opt_3}{}\sumlimits{j \in \bo{o} \cap 
					\pencli{\soln}}{}3\oj{j} 
				+ \sumlimits{j \in \pencli{\opt} \setminus \pencli{\soln}}{} 3\pj{j} \\
				+ \sumlimits{s \in \light}{}\sumlimits{j \in \bs{s}\setminus \pencli{\opt}}{}  6\oj{j} 
			\end{multline}		
			as for
			$\sumlimits{s \in \light}{}\sumlimits{j \in \bs{s}\setminus \pencli{\opt}}{}  (\oj{j} + \oj{\tau(j)} + \sj{\tau(j)} - \sj{j}) = \sumlimits{s \in \light}{}\sumlimits{j \in \bs{s}\setminus \pencli{\opt}}{}2\oj{j}$
			%$j \in \bs{s}\setminus \pencli{\opt}$, we have $(\oj{j} + \oj{\tau(j)} + \sj{\tau(j)} - \sj{j}) = 2\oj{j}$ 
			by the property of $\tau$.
			
			%For $j \in  \pencli{\soln} \cap \pencli{\opt}$ we have $\sj{j} = \oj{j} = \pj{j}$. 
			Adding $\sumlimits{j \in  \pencli{\soln} \cap \pencli{\opt}}{}\pj{j}$ on both the sides and re-arranging, we get
			
			\begin{multline} \label{eq0.4}
				\sumlimits{o \in \opt_1 \cup \opt_2 \cup \opt_3}{}\sumlimits{j \in \bo{o} \setminus \pencli{\soln}}{} \sj{j}   +
				\sumlimits{o \in \opt_1 \cup \opt_2 \cup \opt_3}{}\sumlimits{j \in \bo{o} \cap 
					\pencli{\soln}}{}\pj{j} + 
				\sumlimits{j \in \pencli{\opt} \setminus \pencli{\soln}}{} \sj{j} + \\
				\sumlimits{j \in  \pencli{\soln} \cap \pencli{\opt}}{}\pj{j} \le
				\sumlimits{o \in \opt_1 \cup \opt_2 \cup \opt_3}{}\sumlimits{j \in \bo{o} \setminus \pencli{\soln}}{} 3\oj{j} 
				+ \sumlimits{o \in \opt_1 \cup \opt_2 \cup \opt_3}{}\sumlimits{j \in \bo{o} \cap 
					\pencli{\soln}}{}3\oj{j} \\
				+ \sumlimits{j \in \pencli{\opt} \setminus \pencli{\soln}}{} 3\pj{j} 
				+ \sumlimits{s \in \light}{}\sumlimits{j \in \bs{s}\setminus \pencli{\opt}}{}  6\oj{j} + \sumlimits{j \in  \pencli{\soln} \cap \pencli{\opt}}{}\pj{j}
			\end{multline}			
			
			Rearranging the terms, we get
			\begin{multline} \label{eq0.5}
				\sumlimits{s \in \soln}{}\sumlimits{j \in \bs{s} \setminus \pencli{\opt}}{} \sj{j}  + \sumlimits{j \in \pencli{\opt} \setminus \pencli{\soln}}{} \sj{j} + \sumlimits{j \in  
					\pencli{\soln} \setminus \pencli{\opt}}{}\pj{j} + \sumlimits{j \in  
					\pencli{\soln} \cap \pencli{\opt}}{}\pj{j}
				\le\\
				\sumlimits{o \in \opt_1 \cup \opt_2 \cup \opt_3}{}\sumlimits{j \in \bo{o} \setminus \pencli{\soln}}{} 3\oj{j} 
				+ \sumlimits{j \in \pencli{\opt} \setminus \pencli{\soln}}{} 3\pj{j} \\
				+ \sumlimits{j \in  
					\pencli{\soln} \setminus \pencli{\opt}}{}3\oj{j} 
				+ \sumlimits{j \in  
					\pencli{\soln} \cap \pencli{\opt}}{}\pj{j} 
				+ \sumlimits{s \in \light}{}\sumlimits{j \in \bs{s}\setminus \pencli{\opt} }{}  6\oj{j} 
			\end{multline}
			
			which gives
			\begin{multline} \label{eq0}
				\sumlimits{s \in \soln}{}\sumlimits{j \in \bs{s}}{} \sj{j} 
				% + \sumlimits{j \in \pencli{\opt} \setminus \pencli{\soln}}{} \sj{j} 
				+ \sumlimits{j \in  \pencli{\soln} }{}\pj{j}
				%  + \sumlimits{j \in  \pencli{\soln} \cap \pencli{\opt}}{}\pj{j}
				\le
				\sumlimits{o \in \opt}{}\sumlimits{j \in \bo{o} }{} 3\oj{j} 
				+ \sumlimits{j \in \pencli{\opt}}{} 3\pj{j} 
				%+ \sumlimits{j \in  \pencli{\soln} \setminus \pencli{\opt}}{}3\oj{j} 
				%+ \sumlimits{j \in  \pencli{\soln} \cap \pencli{\opt}}{}\pj{j} 
				+ \sumlimits{s \in \light}{}\sumlimits{j \in \bs{s}\setminus \pencli{\opt} \}}{}  6\oj{j} \\
				\le
				\sumlimits{o \in \opt}{}\sumlimits{j \in \bo{o} }{} 9\oj{j} 
				+ \sumlimits{j \in \pencli{\opt}}{} 3\pj{j} = 9\cost{\opt} 
			\end{multline}
	
			\section{Conclusion}
			In this paper, we presented $((3+\epsilon),8/3)$ factor algorithm for Capacitated $k$-Median Problem  and its penalty version with uniform capacities, using local search heuristic. There is a trade-off between the approximation factor and the cardinality violation between our work and the existing work. Work in~\cite{KPR} is closest to our work as it is the only result for the problem based on local search. We improve upon the results in~\cite{KPR} in terms of cardinality violation.
			
			 It would be interesting to see how these results extend to the problems with non-uniform capacities.

		\bibliography{ref_master_original}
	
\end{document}